\documentclass[letterpaper, 10 pt, conference]{ieeeconf}  

\IEEEoverridecommandlockouts                              
\usepackage{amsmath,amssymb,enumerate,subfigure,multirow,hhline,color}
\usepackage{xcolor}
\usepackage{hyperref}
\usepackage{graphicx}
\usepackage{graphics}
\usepackage[sort]{cite}

\usepackage{algorithm}
\usepackage{algpseudocode}

\usepackage{epsfig,psfrag}
\usepackage{tabularx}
\usepackage{tabulary}

\usepackage[sort]{cite}

\usepackage{hyperref}
\hypersetup{breaklinks=true,colorlinks=true,linkcolor=black,citecolor=black}

\usepackage[noabbrev,capitalise,nameinlink]{cleveref}

\DeclareMathOperator{\argmin}{arg\,min}

\newtheorem{theorem}{Theorem}

\newtheorem{assumption}{Assumption}

\newtheorem{lemma}{Lemma}
\newtheorem{definition}{Definition}
\newtheorem{proposition}{Proposition}

\allowdisplaybreaks

\title{\LARGE \bf
Analysis of Off-Policy Multi-Step TD-Learning with Linear Function Approximation}

\author{Donghwan Lee
\thanks{D. Lee is with the Department of Electrical and Engineering, Korea Advanced Institute of Science and Technology (KAIST), Daejeon, 34141, South Korea {\tt\small
donghwan@kaist.ac.kr}.}
\thanks{This work was supported by Institute of Information communications Technology Planning Evaluation (IITP) grant funded by the Korea government (MSIT)(No.2022-0-00469)}
}

\begin{document}

\maketitle
\thispagestyle{empty}
\pagestyle{empty}

\begin{abstract}
This paper analyzes multi-step TD-learning algorithms within the ``deadly triad'' scenario, characterized by linear function approximation, off-policy learning, and bootstrapping. In particular, we prove that $n$-step TD-learning algorithms converge to a solution as the sampling horizon $n$ increases sufficiently. The paper is divided into two parts. In the first part, we comprehensively examine the fundamental properties of their model-based deterministic counterparts, including projected value iteration, gradient descent algorithms, and the control theoretic approach, which can be viewed as prototype deterministic algorithms whose analysis plays a pivotal role in understanding and developing their model-free reinforcement learning counterparts. In particular, we prove that these algorithms converge to meaningful solutions when $n$ is sufficiently large. Based on these findings, two $n$-step TD-learning algorithms are proposed and analyzed, which can be seen as the model-free reinforcement learning counterparts of the gradient and control theoretic algorithms.
\end{abstract}

\section{Introduction}
Reinforcement learning (RL)~\cite{sutton1998reinforcement} seeks to find an optimal sequence of decisions in unknown systems through experiences. Recent breakthroughs showcase RL algorithms surpassing human performance in various challenging tasks~\cite{mnih2015human,lillicrap2016continuous,heess2015memory,van2016deep,bellemare2017distributional,schulman2015trust,schulman2017proximal}. This success has ignited a surge of interest in RL, both theoretically and experimentally.

Among various algorithms, temporal-difference (TD) learning~\cite{sutton1988learning} stands as a cornerstone of RL, specifically for policy evaluation. Its convergence has been extensively studied over decades~\cite{tsitsiklis1997analysis}. However, a critical challenge emerges within the ``deadly triad'' scenario, characterized by linear function approximation, off-policy learning, and bootstrapping~\cite{sutton1998reinforcement,van2018deep,chen2023target}. In such scenarios, TD-learning can diverge, leading to unreliable value estimates.

Recently, gradient temporal-difference learning (GTD) has been developed and investigated in various studies~\cite{sutton2009convergent,sutton2009fast,ghiassian2020gradient,lee2023new,lim2022backstepping}. This method addresses the deadly triad issue by employing gradient-based schemes. However, the GTD family of algorithms requires somewhat restrictive assumptions about the underlying environment, which constitutes a limitation of the method.

On the other hand, TD-learning is usually implemented within the context of single-step bootstrapping based on a single transition, which is known as single-step TD-learning. These methods can be extended to include multiple time steps, a class of algorithms known as multi-step TD learning, to enhance performance. Recently, multi-step approaches~\cite{sutton1998reinforcement,tsitsiklis1997analysis,chen2021finite,mahmood2017multi,de2018multi,precup2001off,maei2010toward,van2016effective,mandal2023n,schulman2015high}, including $n$-step TD-learning and TD($\lambda$), have become integral to the success of modern deep RL agents, significantly improving performance~\cite{schulman2015high,yuan2019novel,hessel2018rainbow,hernandez2019understanding} in various scenarios. Despite these empirical successes and the growing body of analysis on multi-step RL, to the best of the author's knowledge, the effects and theoretical underpinnings of $n$-step TD-learning have yet to be fully explored.

Motivated by the aforementioned discussions, this paper conducts an in-depth examination of the theoretical foundations necessary to understand the core principles of $n$-step TD-learning methods and their model-based counterparts, which can be viewed as prototype deterministic algorithms whose analysis plays a pivotal role in understanding and developing their model-free RL counterparts. Specifically, we demonstrate that $n$-step TD methods can effectively address the challenges posed by the deadly triad, provided that the sampling horizon size $n$ is sufficiently large. We prove that in this case (sufficiently large $n$), the projected Bellman equation becomes a contraction mapping. This property ensures the convergence of the corresponding TD-learning algorithm towards a useful solution, which we subject to thorough analysis. Moreover, we explore the relationships between the solutions derived from the $n$-step TD methods and those obtained from the projected $n$-step Bellman equation, providing valuable insights into their interconnections.
The paper is divided into two parts: 1) model-based deterministic algorithms and 2) model-free stochastic algorithms.
\begin{enumerate}
\item Model-based deterministic algorithms: This part focuses on three algorithms crucial for the development of multi-step TD-learning algorithms: $n$-step projected value iteration ($n$-PVI), gradient descent algorithms, and dynamical system theoretic algorithms. We show that when the horizon size $n$ is sufficiently large, these algorithms exhibit favorable properties such as contraction mapping, strong convexity, and Schur stability. These properties are essential for understanding and developing $n$-step TD-learning algorithms in the next section.

\item Model-free stochastic algorithms: This part introduces two $n$-step TD-learning algorithms, a naive $n$-step TD-learning and $n$-step GTD, and analyzes their convergence based on the results from the previous section. Again, we prove that for a sufficiently large $n$, these algorithms are guaranteed to converge to useful solutions.
\end{enumerate}

\section{Preliminaries}

\subsection{Notation}
The adopted notation is as follows: ${\mathbb R}$: set of real numbers; ${\mathbb R}^n $: $n$-dimensional Euclidean
space; ${\mathbb R}^{n \times m}$: set of all $n \times m$ real
matrices; $A^T$: transpose of matrix $A$; $A \succ 0$ ($A \prec
0$, $A\succeq 0$, and $A\preceq 0$, respectively): symmetric
positive definite (negative definite, positive semi-definite, and
negative semi-definite, respectively) matrix $A$; $I$: identity matrix with appropriate dimensions; $\lambda_{\min}(A)$ and $\lambda_{\max}(A)$ for any symmetric matrix $A$: the minimum and maximum eigenvalues of $A$; $|{\cal S}|$: cardinality of a finite set $\cal S$.

\subsection{Markov decision process}
A Markov decision process (MDP) is characterized by a quadruple ${\mathcal M}: =
({\cal S},{\mathcal A},P,r,\gamma)$, where ${\mathcal S}$ is a finite
state-space, $\cal A$ is a finite action
space, $P(s'|s,a)$ represents the (unknown)
state transition probability from state $s$ to $s'$ given action
$a$, $r:{\mathcal S}\times {\mathcal A}\times {\mathcal S}\to
{\mathbb R}$ is the reward
function, and $\gamma \in (0,1)$ is the discount factor. In particular, if action
$a$ is selected with the current state $s$, then the state
transits to $s'$ with probability $P(s'|s,a)$ and incurs a
reward $r(s,a,s')$.  For convenience, we consider a deterministic reward function and simply write $r(s_k,a_k ,s_{k + 1}) =:r_{k+1},k \in \{ 0,1,\ldots \}$.

The stochastic policy is a map $\pi:{\mathcal S} \times
{\mathcal A}\to [0,1]$ representing the probability, $\pi(a|s)$, of selecting action $a$ at the current state $s$, $P^\pi$ denotes the state transition probability matrix under policy $\pi$, and $d^{\pi}:{\mathcal S} \to {\mathbb R}$ denotes the stationary probability distribution of the state $s\in {\mathcal S}$ under $\pi$. We also define
$R^\pi(s)$ as the expected reward given the policy $\pi$ and the current state $s$. The infinite-horizon discounted value function with policy $\pi$ is $V^\pi(s):={\mathbb E} \left[ \left. \sum_{k = 0}^\infty {\gamma
^k r(s_k,a_k,s_{k+1})} \right|s_0=s \right]$, where ${\mathbb E}$ stands for the expectation taken with respect to the state-action trajectories under $\pi$. Given pre-selected basis (or feature) functions $\phi_1,\ldots,\phi_m:{\mathcal S}\to {\mathbb R}$, the matrix, $\Phi \in {\mathbb R}^{|{\mathcal S}| \times m}$, called the feature matrix, is defined as a matrix whose $s$-th row vector is $\phi(s):=\begin{bmatrix} \phi_1(s) &\cdots & \phi_m(s) \end{bmatrix}$. Throughout the paper, we assume that $\Phi \in {\mathbb R}^{|{\mathcal S}| \times m}$ is a full column rank matrix. The policy evaluation problem is the problem of estimating $V^{\pi}$ given a policy $\pi$.

\subsection{Review of GTD algorithm}
In this section, we briefly review the gradient temporal difference (GTD) learning developed in~\cite{sutton2009convergent}, which tries to solve the policy evaluation problem. Roughly speaking, the goal of the policy evaluation is to find the weight vector $\theta$ such that $\Phi \theta$ approximates the true value function $V^{\pi}$. This is typically done by minimizing the so-called {\em mean-square projected Bellman error} loss function~\cite{sutton2009convergent,sutton2009fast}
\begin{align}
&\min_{\theta\in {\mathbb R}^q} {\rm MSPBE}(\theta):= \frac{1}{2}\|
\Pi (R^{\pi} + \gamma P^{\pi} \Phi \theta)-\Phi \theta \|_{D^{\beta}}^2.\label{eq:4}
\end{align}
where $R^\pi \in {\mathbb R}^{|{\mathcal S}|}$ is a vector enumerating all $R^\pi(s), s\in {\mathcal S}$, $D^{\beta}$ is a diagonal matrix with positive diagonal elements $d^{\beta}(s),s\in {\mathcal S}$, and $\|x\|_D:=\sqrt{x^T Dx}$ for any positive-definite $D$. Here, $d^{\beta}$ can be any state visit distribution under the behavior policy $\beta$ such that $d^{\beta}(s)>0,\forall s\in {\mathcal S}$. Moreover, $\Pi$ is the projection onto the range space of $\Phi$, denoted by ${\cal R}(\Phi)$: $\Pi(x):=\argmin_{x'\in {\cal R}(\Phi)}
\|x-x'\|_{D^{\beta}}^2$. The projection can be performed by the matrix
multiplication: we write $\Pi(x):=\Pi x$, where $\Pi:=\Phi(\Phi^T
D^{\beta} \Phi)^{-1}\Phi^T D^{\beta} \in {\mathbb R}^{|{\cal S}|\times |{\cal S}|}$.

Note that minimizing the objective means minimizing the error of the projected Bellman equation (PBE) $\Phi \theta  = \Pi (R^\pi   + \gamma P^\pi  \Phi \theta )$ with respect to $ \|\cdot\|_{D^\beta}$. Moreover, note that in the objective of~\eqref{eq:4}, $d^{\beta}$ depends on the behavior policy, $\beta$, while $P^{\pi}$ and $R^{\pi}$ depend on the target policy, $\pi$, that we want to evaluate. This structure allows us to obtain an off-policy learning algorithm through the importance sampling~\cite{precup2001off} or sub-sampling techniques~\cite{sutton2009convergent}. Throughout the paper, we frequently adopt the following standard assumption.
\begin{assumption}\label{assumption:2}
$\Phi ^T D^{\beta} (\gamma P^\pi   - I)\Phi $ is nonsingular, where $I$ denotes the identity matrix with an appropriate dimension.
\end{assumption}
Note that~\cref{assumption:2} is common in the literature, and is adopted in~\cite{sutton2009convergent,sutton2009fast,ghiassian2020gradient,lee2023new} for convergence of GTD algorithms. Some properties related to~\eqref{eq:4} are summarized below for convenience and completeness.
\begin{lemma}[{\cite{lee2023new}}]\label{lemma:2}
Suppose that~\cref{assumption:2} holds. Then, the following statements hold true:
\begin{enumerate}
\item A solution of~\eqref{eq:4} exists, and is unique.

\item The solution of~\eqref{eq:4} is given by
\begin{align}
\theta ^*:=-(\Phi ^T D^{\beta} (\gamma P^\pi   - I)\Phi )^{ - 1} \Phi ^T D^{\beta} R^\pi.\label{eq:theta-star}
\end{align}
\end{enumerate}
\end{lemma}

Based on this objective function,~\cite{sutton2009fast} developed GTD2.

\section{Multi-step projected Bellman operator}\label{sec:dynamic-programming}

Let us consider the $n$-step Bellman operator~\cite{sutton1988learning}
\begin{align*}
{T^n}(x):=& {R^\pi } + \gamma {P^\pi }{R^\pi } +  \cdots  + {\gamma ^{n - 1}}{({P^\pi })^{n - 1}}{R^\pi }\\
& + {\gamma ^n}{({P^\pi })^n}x.
\end{align*}

Then, the corresponding projected $n$-step Bellman operator ($n$-PBO) is given by $\Pi {T^n}$. Based on this, the corresponding $n$-step projected value iteration ($n$-PVI) is given by
\begin{align}
\Phi {\theta _{k + 1}} = \Pi {T^n}(\Phi {\theta _k}),\quad k \in \{ 0,1, \ldots \} ,\quad {\theta _0} \in {\mathbb R}^m\label{eq:projected-VI}
\end{align}

Note that at each iteration $k$, $\theta _{k + 1}$ can be uniquely determined given $\theta _k$ because $\Pi {T^n}(\Phi {\theta _k}) \in {\cal R}(\Phi)$, and the unique solution solves
\begin{align*}
\Phi \theta  = \Pi {T^n}(\Phi {\theta _k}),
\end{align*}
and is given by
\begin{align*}
{\theta _{k + 1}} = {({\Phi ^T}D^\beta\Phi )^{ - 1}}{\Phi ^T}D^\beta {T^n}(\Phi {\theta _k}).
\end{align*}

It is important to note that $\Pi \in {\mathbb R}^{|{\cal S}|\times |{\cal S}|}$ is a projection onto the column space of the feature matrix $\Phi$ with respect to the weighted norm ${\left\|  \cdot  \right\|_{{D^\beta }}}$, and satisfies the nonexpansive mapping property ${\left\| {\Pi x - \Pi y} \right\|_{{D^\beta }}} \le {\left\| {x - y} \right\|_{{D^\beta }}}$ with respect to ${\left\|  \cdot  \right\|_{{D^\beta }}}$.
On the other hand, for the Bellman operator $T^n$, we can consider the two cases:
\begin{enumerate}
\item on-policy case $\beta = \pi$,
\item off-policy case $\beta \neq \pi$.
\end{enumerate}

In the on-policy case $\beta = \pi$, it can be easily proved that $T^n$ is a contraction mapping with respect to the norm ${\left\|  \cdot  \right\|_{{D^\beta }}}$ with the contraction factor $\gamma^n$.
\begin{lemma}
The mapping $T^n$ satisfies
\begin{align*}
{\left\| {{T^n}(x) - {T^n}(y)} \right\|_{{D^\pi }}} \le {\gamma ^n}{\left\| {x - y} \right\|_{{D^\pi }}},\quad \forall x,y \in {\mathbb R}^{|S|}.
\end{align*}
\end{lemma}
\begin{proof}
The proof can be easily done by following the main ideas of~\cite[Lemma~4]{tsitsiklis1997analysis}, and omitted here for brevity.
\end{proof}

Therefore, $n$-PBO$, \Pi {T^n}$, is also a contraction with the factor $\gamma^n$.
\begin{lemma}
The mapping $\Pi T^n$ satisfies
\begin{align*}
{\left\| {{\Pi T^n}(x) - {\Pi T^n}(y)} \right\|_{{D^\pi }}} \le {\gamma ^n}{\left\| {x - y} \right\|_{{D^\pi }}},\quad \forall x,y \in {\mathbb R}^{|S|}.
\end{align*}
\end{lemma}

In conclusion, by Banach fixed point theorem, $n$-PVI in~\eqref{eq:projected-VI} converges to its unique fixed point because $n$-PVO $\Pi {T^n}$ is a contraction with respect to ${\left\|  \cdot  \right\|_{{D^\beta }}}$.

On the other hand, in the off-policy case $\beta \neq \pi$, $T^n$ is no more a contraction mapping with respect to ${\left\|  \cdot  \right\|_{{D^\beta }}}$, and so is $\Pi T^n$. Therefore, $n$-PVI in~\eqref{eq:projected-VI} may not converge in some cases.
However, it can be proved that for a sufficiently large $n$, $\Pi T^n$ becomes contraction with respect to the different norm ${\left\|  \cdot  \right\|_\infty}$. In this paper, we will formally and rigorously address this property. Before proceeding, the following assumptions are introduced.
\begin{assumption}
Throughout the paper, we consider the off-policy scenario, $\beta \neq \pi$.
Moreover, we assume that ${\left\| \Pi  \right\|_\infty } \ge 1$ because otherwise, $\Pi T^n$ is trivially a contraction with respect to ${\left\|  \cdot  \right\|_\infty}$.
\end{assumption}

Next, we establish the contraction property.
\begin{theorem}\label{thm:contraction1}
There exists a positive integer $n^*$ such that the mapping $\Pi T^n$ is a contraction with respect to ${\left\|  \cdot  \right\|_\infty}$ for all $n \geq n^*$. One such $n^*$ is given by
\begin{align*}
{n^*} = \left\lceil {\frac{{\ln (\left\| \Pi  \right\|_\infty ^{ - 1})}}{{\ln (\gamma )}}} \right\rceil  + 1,
\end{align*}
where $\left\lceil  \cdot  \right\rceil$ stands for the ceiling function.
\end{theorem}
\begin{proof}
Noting that
\begin{align*}
{\left\| {\Pi {T^n}(x) - \Pi {T^n}(y)} \right\|_\infty } =& {\gamma ^n}{\left\| {\Pi ({{({P^\pi })}^n}x - {{({P^\pi })}^n}y)} \right\|_\infty }\\
\le& {\gamma ^n}{\left\| \Pi  \right\|_\infty }{\left\| {{{({P^\pi })}^n}(x - y)} \right\|_\infty }\\
\le& {\gamma ^n}{\left\| \Pi  \right\|_\infty }{\left\| {{{({P^\pi })}^n}} \right\|_\infty }{\left\| {x - y} \right\|_\infty }\\
=& {\gamma ^n}{\left\| \Pi  \right\|_\infty }{\left\| {x - y} \right\|_\infty }
\end{align*}
Therefore, for a sufficiently large $n^*$, we have ${\gamma ^n}{\left\| \Pi  \right\|_\infty } < 1$ for all $n \ge n^*$, which implies that $\Pi T^n$ is a contraction mapping with respect to ${\left\|  \cdot  \right\|_\infty}$. In particular, ${\gamma ^n}{\left\| \Pi  \right\|_\infty } < 1$ is equivalent to $n\ln (\gamma ) < \ln (\left\| \Pi  \right\|_\infty ^{ - 1})$, or equivalently, $n > \frac{{\ln (\left\| \Pi  \right\|_\infty ^{ - 1})}}{{\ln (\gamma )}}$. Taking the ceiling function on the left-hand side, a sufficient condition is $n > \left\lceil {\frac{{\ln (\left\| \Pi  \right\|_\infty ^{ - 1})}}{{\ln (\gamma )}}} \right\rceil$. This completes the proof.
\end{proof}

In the next theorem, we establish a connection between the contraction property of $\Pi {T^n}$ and the nonsingularity of ${\Phi ^T}{D^\beta }(I - {\gamma ^n}{({P^\pi })^n})\Phi$, which plays an important role throughout the paper.
\begin{theorem}\label{thm:contraction-property1}
${\Phi ^T}{D^\beta }(I  - {\gamma ^n}{({P^\pi })^n})\Phi$ is nonsingular if $\Pi {T^n}$ is a contraction.
\end{theorem}
\begin{proof}
Suppose that $\Pi {T^n}$ is a contraction. Then, it admits a unique fixed point $\theta^*$ satisfying $\Phi \theta^* = \Pi {T^n} (\Phi \theta^*)$, which is equivalent to
\begin{align*}
&{\Phi ^T}D^\beta({R^\pi } + \gamma {P^\pi }{R^\pi } +  \cdots  + {\gamma ^{n - 1}}{({P^\pi })^{n - 1}}{R^\pi })\\
 =& {\Phi ^T}D^\beta (I - {\gamma ^n}{({P^\pi })^n})\Phi \theta
\end{align*}
For the above equation to have a unique solution, ${\Phi ^T}{D^\beta }(I - {\gamma ^n}{({P^\pi })^n})\Phi$ should be nonsingular. This completes the proof.
\end{proof}

Therefore, for any $n \geq n^*$, $n$-PVI in~\eqref{eq:projected-VI} converges to the unique fixed point, denoted by $\theta _*^n$, which satisfies
\begin{align}
\Phi \theta _*^n = \Pi {T^n}(\Phi \theta _*^n).\label{eq:1}
\end{align}

The unique fixed point is given as follows.
\begin{lemma}\label{lemma:1}
Suppose that $n \geq n^*$ so that $\Pi {T^n}$ is a contraction with respect to ${\left\|  \cdot  \right\|_\infty}$.
Then, the unique fixed point of $\Pi {T^n}$, denoted by $\theta _*^n$, is given by
\begin{align}
\theta_*^n =& {\left[ {{\Phi ^T}{D^\beta }(I - {\gamma ^n}{{({P^\pi })}^n})\Phi } \right]^{ - 1}}{\Phi ^T}{D^\beta }\nonumber\\
&\times ({R^\pi } + \gamma {P^\pi }{R^\pi } +  \cdots  + {\gamma ^{n - 1}}{({P^\pi })^{n - 1}}{R^\pi }).\label{eq:solution1}
\end{align}
\end{lemma}
\begin{proof}
The fixed point equation in~\eqref{eq:1} can be rewritten by
\begin{align*}
&{\Phi ^T}D^\beta({R^\pi } + \gamma {P^\pi }{R^\pi } +  \cdots  + {\gamma ^{n - 1}}{({P^\pi })^{n - 1}}{R^\pi })\\
 =& {\Phi ^T}D^\beta(I - {\gamma ^n}{({P^\pi })^n})\Phi \theta _*^n
\end{align*}
By~\cref{thm:contraction-property1}, ${\Phi ^T}D^\beta(I - {\gamma ^n}{({P^\pi })^n})\Phi$ is nonsingular, and hence, the unique solution is given by~\eqref{eq:solution1}.
\end{proof}

The convergence speed of $n$-PVI in~\eqref{eq:projected-VI} is given below.
\begin{theorem}(Convergence)\label{thm:convergence1}
Suppose that $n \geq n^*$ so that $\Pi {T^n}$ is a contraction with respect to ${\left\|  \cdot  \right\|_\infty}$.
$n$-PVI in~\eqref{eq:projected-VI} satisfies
\begin{align*}
{\left\| {\Phi {\theta _k} - \Phi \theta _*^n} \right\|_\infty } \le {({\gamma ^n}{\left\| \Pi  \right\|_\infty })^k}{\left\| {\Phi {\theta _0} - \Phi \theta _*^n} \right\|_\infty }.
\end{align*}
\end{theorem}
\begin{proof}
It is straightforward from the contraction property in~\Cref{thm:contraction1}.
\end{proof}

The results in~\cref{lemma:1} tell us that the solution $\theta_*^n$ of $n$-PBE~\eqref{eq:1} varies according to $n$. A natural question that arises is: what is the significance of the solution $\theta_*^n$, and what constitutes the true solution that we want to find. To answer this fundamental question, the desired true solution is first defined below.
\begin{definition}[Optimal solution]\label{def:1}
The optimal solution, denoted by $\theta _*^\infty$, is defined as a vector $\theta _*^\infty$ that satisifes
\begin{align}
\Phi \theta _*^\infty  = \Pi {V^\pi },\label{eq:optimal-solution}
\end{align}
where
\begin{align*}
{V^\pi }: = \sum\limits_{k = 0}^\infty  {{\gamma ^k}{{({P^\pi })}^k}{R^\pi }}
\end{align*}
is the true value function.
\end{definition}

Note that the solution in~\eqref{eq:optimal-solution} can be interpreted as the least-square solution of $\theta _*^\infty  = \arg {\min _{\theta  \in {\mathbb R}^n}}f(\theta )$, where
\begin{align*}
f(\theta ) = \frac{1}{2}\left\| {{V^\pi } - \Phi \theta } \right\|_{{D^\beta }}^2.
\end{align*}

In terms of~\cref{def:1}, the next natural question is regarding the relevance of $\theta_*^n$ in comparison to the true optimal solution $\theta_*^\infty$ and the true value function $V^\pi$.
\begin{theorem}
Suppose that $n \geq n^*$ so that $\Pi {T^n}$ is a contraction with respect to ${\left\|  \cdot  \right\|_\infty}$.
Then, we have
\begin{align}
{\left\| {\Phi \theta^n_* - {V^\pi }} \right\|_\infty } \le \frac{1}{{1 - {\gamma ^n}{{\left\| \Pi  \right\|}_\infty }}}{\left\| {\Pi {V^\pi } - {V^\pi }} \right\|_\infty }\label{eq:2}
\end{align}
and
\begin{align}
{\left\| {\Phi \theta^n_* - \Phi {\theta_*^\infty}} \right\|_\infty } \le \frac{{{\gamma ^n}{{\left\| \Pi  \right\|}_\infty }}}{{1 - {\gamma ^n}{{\left\| \Pi  \right\|}_\infty }}}{\left\| {\Pi {V^\pi } - {V^\pi }} \right\|_\infty }\label{eq:3}
\end{align}
\end{theorem}
\begin{proof}
By hypothesis, $\Pi {T^n}$ is a contraction, which means that there exists a unique solution $\theta ^n_*$ satisfying $n$-PBE~\eqref{eq:1}, which can be rewritten by
\begin{align*}
\Pi {T^n}(\Phi \theta^n_*) - {V^\pi } = \Phi \theta^n_* - {V^\pi }.
\end{align*}

The left-hand side can be written as
\begin{align*}
&\Phi \theta^n_* - {V^\pi }\nonumber\\
=&\Pi {T^n}(\Phi \theta^n_*) - {V^\pi }\nonumber\\
=& \Pi {T^n}(\Phi \theta^n_*) - {V^\pi } - \Pi {V^\pi } + \Pi {V^\pi }\nonumber\\
=& \Pi ({T^n}(\Phi \theta^n_*) - {V^\pi }) + \Pi {V^\pi } - {V^\pi }\nonumber\\
=& \Pi \left( {{V^\pi } - \sum\limits_{k = n}^\infty  {{\gamma ^k}{{({P^\pi })}^k}{R^\pi }}  + {\gamma ^n}{{({P^\pi })}^n}\Phi \theta^n_* - {V^\pi }} \right)\nonumber\\
& + \Pi {V^\pi } - {V^\pi }\nonumber\\
=& \Pi ({\gamma ^n}{({P^\pi })^n}\Phi \theta^n_* - {\gamma ^n}{({P^\pi })^n}{V^\pi }) + \Pi {V^\pi } - {V^\pi }.
\end{align*}

Next, taking the norm ${\left\|  \cdot  \right\|_\infty}$ on both sides of the above inequality leads to
\begin{align*}
&{\left\| {\Phi \theta^n_* - {V^\pi }} \right\|_\infty }\\
=& {\left\| {\Pi ({\gamma ^n}{{({P^\pi })}^n}\Phi \theta^n_* - {\gamma ^n}{{({P^\pi })}^n}{V^\pi }) + \Pi {V^\pi } - {V^\pi }} \right\|_\infty }\\
\le& {\left\| {\Pi ({\gamma ^n}{{({P^\pi })}^n}\Phi \theta^n_* - {\gamma ^n}{{({P^\pi })}^n}{V^\pi })} \right\|_\infty } + {\left\| {\Pi {V^\pi } - {V^\pi }} \right\|_\infty }\\
\le& {\gamma ^n}{\left\| \Pi  \right\|_\infty }{\left\| {{{({P^\pi })}^n}} \right\|_\infty }{\left\| {\Phi \theta^n_* - {V^\pi }} \right\|_\infty }\\
&\times {\left\| {{{({P^\pi })}^n}\Phi \theta^n_* - {{({P^\pi })}^n}{V^\pi }} \right\|_\infty }\\
& + {\left\| {\Pi {V^\pi } - {V^\pi }} \right\|_\infty }\\
=& {\gamma ^n}{\left\| \Pi  \right\|_\infty }{\left\| {\Phi \theta^n_* - {V^\pi }} \right\|_\infty } + {\left\| {\Pi {V^\pi } - {V^\pi }} \right\|_\infty },
\end{align*}
which yields
\begin{align*}
(1 - {\gamma ^n}{\left\| \Pi  \right\|_\infty }){\left\| {\Phi \theta^n_* - {V^\pi }} \right\|_\infty } \le {\left\| {\Pi {V^\pi } - {V^\pi }} \right\|_\infty }.
\end{align*}
By hypothesis, $n \geq n^*$ implies that $1 - {\gamma ^n}{\left\| \Pi  \right\|_\infty } > 0$ holds. Therefore, the last inequality leads to~\eqref{eq:2}.

Similarly, combining~\eqref{eq:1} and~\eqref{eq:optimal-solution} yields
\begin{align*}
&\Phi (\theta _*^n - \theta _*^\infty )\\
=& \Pi ({T^n}(\Phi \theta _*^n) - {V^\pi })\\
=& \Pi \left( {{V^\pi } - \sum\limits_{k = n}^\infty  {{\gamma ^k}{{({P^\pi })}^k}{R^\pi }}  + {\gamma ^n}{{({P^\pi })}^n}\Phi \theta _*^n - {V^\pi }} \right)\\
=& \Pi \left( {{\gamma ^n}{{({P^\pi })}^n}\Phi \theta _*^n - \sum\limits_{k = n}^\infty  {{\gamma ^k}{{({P^\pi })}^k}{R^\pi }} } \right)\\
 =& {\gamma ^n}\Pi ({({P^\pi })^n}\Phi \theta _*^n - {({P^\pi })^n}{V^\pi }).
\end{align*}
Now, taking the norm ${\left\|  \cdot  \right\|_\infty}$ on both sides of the above inequality leads to
\begin{align*}
&{\left\| {\Phi (\theta _*^n - \theta _*^\infty )} \right\|_\infty }\\
=& {\left\| {\Pi {{({P^\pi })}^n}\Phi \theta _*^n - \Pi {{({P^\pi })}^n}{V^\pi }} \right\|_\infty }\\
\le& {\gamma ^n}{\left\| \Pi  \right\|_\infty }{\left\| {\Phi \theta _*^n - {V^\pi }} \right\|_\infty }\\
\le& {\gamma ^n}{\left\| \Pi  \right\|_\infty }\frac{1}{{1 - {\gamma ^n}{{\left\| \Pi  \right\|}_\infty }}}{\left\| {\Pi {V^\pi } - {V^\pi }} \right\|_\infty },
\end{align*}
where the second inequality comes from~\eqref{eq:2}. This completes the proof.
\end{proof}

The inequality in~\eqref{eq:2} tells us an error bound between $\Phi \theta^n_*$ and the true value function ${V^\pi }$.
Moreover,~\eqref{eq:3} gives an error bound between $\Phi \theta^n_*$ and the true optimal solution $\Phi \theta^\infty_*$.
One can observe that the second bound in~\eqref{eq:3} vanish as $n \to \infty$, which is reasonable because as $n \to \infty$, the left-hand side of $n$-PVI in~\eqref{eq:1} becomes identical to the left-hand side of~\eqref{eq:optimal-solution} which is the projection of the true value function without the bootstrapping. On the other hand, the first bound in~\eqref{eq:2} does not vanish as $n \to \infty$. This is because even though $n \to \infty$, there still remains a fundamental error between the true value function and the estimated value function in the linear function class, which cannot be overcame.

Until now, we have studied some properties of $n$-PBO and the corresponding $n$-PVI in~\eqref{eq:1}. These properties play important roles for the development of the corresponding model-free algorithms. In the next sections, we will study some alternative approaches based on gradients to solve the policy evaluation problem.

\section{Gradient operator~I}\label{sec:gradient1}

Let us consider the objective function, called the $n$-ste[ mean-square projected Bellman error ($n$-MSPBE) loss function~\cite{sutton2009convergent,sutton2009fast}
\begin{align}
f(\theta ) = \frac{1}{2}\left\| {\Pi {T^n}(\Phi \theta ) - \Phi \theta } \right\|_{D^\beta}^2\label{eq:objective1}
\end{align}
and the corresponding optimization problem
\begin{align}
{\bar \theta ^n}: = \argmin _{\theta  \in {\mathbb R}^m}f(\theta ).\label{eq:opt1}
\end{align}

The objective function~\eqref{eq:objective1} is a popular class of objective functions that plays an important role in the analysis of TD-learning and GTD algorithms~\cite{sutton2009convergent,sutton2009fast}. The corresponding gradient is given by
\begin{align}
{\nabla _\theta }f(\theta ) =& {\Phi ^T}{({\gamma ^n}{({P^\pi })^n} - I)^T}{\Pi ^T}{D^\beta }(\Pi {T^n}(\Phi \theta ) - \Phi \theta),\nonumber\\
=& {\Phi ^T}{({\gamma ^n}{({P^\pi })^n} - I)^T}{D^\beta }\Phi {({\Phi ^T}{D^\beta }\Phi )^{ - 1}}\nonumber\\
&\times {\Phi ^T}{D^\beta }(\Pi {T^n}(\Phi \theta ) - \Phi \theta ),\label{eq:gradient1}
\end{align}
and the Hessian is
\begin{align*}
\nabla _\theta ^2f(\theta )
=& {\Phi ^T}{({\gamma ^n}{({P^\pi })^n} - I)^T}{D^\beta }\Phi {({\Phi ^T}D^\beta \Phi )^{ - 1}}\\
&\times {\Phi ^T}{D^\beta }({\gamma ^n}{({P^\pi })^n} - I)\Phi\\
=:& \Gamma
\end{align*}

In this section, we consider the gradient descent algorithm
\begin{align}
{\theta _{k + 1}} = {\theta _k} - \alpha {\nabla _\theta }f({\theta _k}),\quad k \in \{ 0,1, \ldots \} ,\quad {\theta _0} \in {\mathbb R}^m,\label{eq:GD-algo1}
\end{align}
which can be an alternative to the dynamic programming algorithm in~\eqref{eq:projected-VI}. The main reason we consider this deterministic gradient descent algorithm is that it can be potentially applied for the development of model-free RLs such as the GTD algorithms~\cite{sutton2009convergent,sutton2009fast,lee2023new} and the residual gradient algorithms~\cite{baird1995residual}.

In the sequel, we will investigate some important properties of the objective function~\eqref{eq:objective1} and its gradient~\eqref{eq:gradient1} such as the convexity and solution analysis. The first natural question is whether or not the objective function~\eqref{eq:objective1} is convex. Since $\Gamma \succeq 0$, it is indeed true.
\begin{theorem}[Convexity]
For any $n \geq 1$, the objective function~\eqref{eq:objective1} is convex.
\end{theorem}

Therefore, the typical gradient descent algorithm shown below can be applied to find a global solution of~\eqref{eq:opt1}, $\bar \theta ^n$, with sublinear convergence rates, which are also a stationary point satisfying
\begin{align}
{\nabla _\theta }f({{\bar \theta }^n}) = 0.\label{eq:stationary1}
\end{align}

However, when $n=1$, without~\cref{assumption:2}, there is no guarantee that the stationary point satisfying~\eqref{eq:stationary1} is identical to the unique fixed point, $\theta^n_*$, of $n$-PBE in~\eqref{eq:1} because ${\Phi ^T}{({\gamma ^n}{P^\pi} - I)^T}{D^\beta }\Phi$ has a non-trivial null space. To overcome this issue, we can consider larger $n$.
In particular, when $n$ is sufficiently large, the objective function~\eqref{eq:objective1} is strongly convex.
\begin{theorem}[Strong convexity]\label{thm:convexity1}
If $n \geq n^*$ so that $\Pi {T^n}$ is a contraction, then the objective function~\eqref{eq:objective1} is $\mu$-strongly convex with $\mu = {\lambda _{\min }}(\Gamma )$.
\end{theorem}
\begin{proof}
By~\cref{thm:contraction-property1}, ${\Phi ^T}{D^\beta }({\gamma ^n}{({P^\pi })^n} - I)\Phi $ is nonsingular.
Then, the Hessian of~\eqref{eq:objective1} is positive definite because ${\Phi ^T}D^\beta\Phi  \succ 0 \Leftrightarrow {({\Phi ^T}D^\beta\Phi )^{ - 1}} \succ 0 \Leftrightarrow {\Phi ^T}{({\gamma ^n}{({P^\pi })^n} - I)^T}D^\beta \Phi {({\Phi ^T}D^\beta \Phi )^{ - 1}}{\Phi ^T}D^\beta ({\gamma ^n}{({P^\pi })^n} - I)\Phi  \succ 0$. Therefore,~\eqref{eq:objective1} is strongly convex, and the coefficient is given by $\mu = {\lambda _{\min }}(\Gamma )$~\cite[Thm.~2.1.11]{nesterov2018lectures}. This completes the proof.
\end{proof}

It can be shown that when strongly convex, the unique global optimal solution of~\eqref{eq:opt1} satisfying~\eqref{eq:stationary1} is identical to the unique fixed point, $\theta^n_*$, of $n$-PBE in~\eqref{eq:1} without~\cref{assumption:2}.
\begin{theorem}[Stationary point]
Suppose that $n \geq n^*$ so that $\Pi {T^n}$ is a contraction. Then, the unique stationary point, ${\bar \theta }^n$, satisfying~\eqref{eq:stationary1} is the unique fixed point, $\theta^n_*$, of $n$-PBE in~\eqref{eq:1}.
\end{theorem}
\begin{proof}
The stationary point satisfies
\begin{align*}
{\nabla _\theta }f({{\bar \theta }^n}) =& {\Phi ^T}{({\gamma ^n}{({P^\pi })^n} - I)^T}D^\beta \Phi {({\Phi ^T}D^\beta \Phi )^{ - 1}}\\
&\times {\Phi ^T}D^\beta (\Pi {T^n}(\Phi {{\bar \theta }^n}) - \Phi {{\bar \theta }^n})\\
 =& 0.
\end{align*}
Since ${\Phi ^T}{({\gamma ^n}{({P^\pi })^n} - I)^T}D^\beta \Phi$ is nonsingular by~\cref{thm:contraction-property1}, the above equality is equivalent to
\begin{align*}
{({\Phi ^T}D^\beta \Phi )^{ - 1}}{\Phi ^T}D^\beta (\Pi {T^n}(\Phi {{\bar \theta }^n}) - \Phi {{\bar \theta }^n}) = 0.
\end{align*}
Next, multiplying both sides of the above equality by $\Phi$ leads to $\Pi (\Pi {T^n}(\Phi {{\bar \theta }^n}) - \Phi {{\bar \theta }^n}) = \Pi {T^n}(\Phi {{\bar \theta }^n}) - \Phi {{\bar \theta }^n} = 0$, which is the fixed point equation of PBE. This completes the proof.
\end{proof}

\begin{theorem}[Lipschitz continuity]\label{thm:Lipschitz1}
Suppose that $n \geq n^*$ so that $\Pi {T^n}$ is a contraction. Then, the gradient in~\eqref{eq:gradient1} is $L$-Lipschitz continuous, i.e.,
\begin{align*}
{\left\| {{\nabla _\theta }f(x) - {\nabla _\theta }f(y)} \right\|_2} \le L{\left\| {x - y} \right\|_2},
\end{align*}
where $L = {\left\| \Gamma  \right\|_2}$.
\end{theorem}
\begin{proof}
Noting the gradient in~\eqref{eq:gradient1}, simple calculations leads to the desired conclusion.
\end{proof}

\begin{theorem}[{\cite[Thm.~2.1.15]{nesterov2018lectures}}]\label{thm:convergence2}
Suppose that $n \geq n^*$ so that $\Pi {T^n}$ is a contraction.
The gradient descent algorithm in~\eqref{eq:GD-algo1} with the constant step-size $\alpha  = \frac{2}{{\mu  + L}}$ satisfies
\begin{align*}
\left\| {{\theta _k} - \theta _*^n} \right\|_2^2 \le {\left( {\frac{{L - \mu }}{{L + \mu }}} \right)^{2k}}\left\| {{\theta _0} - \theta _*^n} \right\|_2^2,
\end{align*}
where $c$ and $L$ such that $L \ge \mu$ are defined in~\cref{thm:convexity1} and~\cref{thm:Lipschitz1}, respectively.
\end{theorem}
\begin{proof}
The proof follows that of~\cite[Thm.~2.1.15]{nesterov2018lectures}.
\end{proof}

Note that the parameters $c$ and $L$, defined in~\cref{thm:convexity1} and~\cref{thm:Lipschitz1}, respectively, are not unique. Therefore, one can adjust them so that the condition $L \ge \mu$ is satisfied.

\section{Gradient operator~II}\label{sec:gradient2}

In this section, let us consider the different objective function
\begin{align}
f(\theta ) = \frac{1}{2}\left\| {{\Phi ^T}{D^\beta}({T^n}(\Phi \theta ) - \Phi \theta )} \right\|_{D^\beta}^2,\label{eq:objective2}
\end{align}
and the corresponding optimization problem
\begin{align}
{\bar \theta ^n}: = \argmin _{\theta  \in {\mathbb R}^m}f(\theta ).\label{eq:opt2}
\end{align}

The objective function in~\eqref{eq:objective2} is different from~\eqref{eq:objective1}, and has been introduced in~\cite{lee2023new}. Similar to MSPBE in~\eqref{eq:objective1},~\eqref{eq:objective2} can be used to derive GTD algorithms.
The corresponding gradient is given by
\begin{align}
{\nabla _\theta }f(\theta )=& {\Phi ^T}{({\gamma ^n}{({P^\pi })^n} - I)^T}{D^\beta }\nonumber\\
& \times \Phi {D^\beta }{\Phi ^T}{D^\beta }({T^n}(\Phi \theta ) - \Phi \theta),\label{eq:gradient2}
\end{align}
and the Hessian is
\begin{align*}
&\nabla _\theta ^2f(\theta )\\
 =& {\Phi ^T}{({\gamma ^n}{({P^\pi })^n} - I)^T}{D^\beta }\Phi {D^\beta }{\Phi ^T}{D^\beta }({\gamma ^n}{({P^\pi })^n} - I)\Phi \\
=& :\Gamma.
\end{align*}

Similar to the previous section, we consider the gradient descent algorithm in~\eqref{eq:GD-algo1}.
Moreover, we investigate properties of the objective function~\eqref{eq:objective2} following similar steps as in the previous section.
\begin{theorem}[Convexity]
For any $n \geq 1$, the objective function~\eqref{eq:objective2} is convex.
\end{theorem}

Therefore, the gradient descent algorithm in~\eqref{eq:GD-algo1} can find a global optimal solution of~\eqref{eq:opt2}, which is also a stationary point satisfying
\begin{align}
{\nabla _\theta }f({{\bar \theta }^n}) = 0.\label{eq:stationary2}
\end{align}

However, when $n=1$, without~\cref{assumption:2}, there is no guarantee that the stationary point satisfying~\eqref{eq:stationary2} is identical to the unique fixed point, $\theta^n_*$, of $n$-PBE in~\eqref{eq:1}.
To overcome this issue, we can consider larger $n$.
In particular, when $n$ is sufficiently large, the objective function~\eqref{eq:objective2} is strongly convex.
\begin{theorem}[Strong convexity]\label{thm:convexity2}
If $n \geq n^*$ so that $\Pi {T^n}$ is a contraction, then the objective function~\eqref{eq:objective2} is $\mu$-strongly convex with $\mu = {\lambda _{\min }}(\Gamma )$.
\end{theorem}
\begin{proof}
First of all, $\Gamma \succeq 0$ in general, and hence, \eqref{eq:objective1} is convex.
By~\cref{thm:contraction-property1}, ${\Phi ^T}{D^\beta }({\gamma ^n}{({P^\pi })^n} - I)\Phi $ is nonsingular.
Then, the Hessian of~\eqref{eq:objective2} is positive definite because ${D^\beta } \succ 0 \Leftrightarrow {\Phi ^T}{({\gamma ^n}{({P^\pi })^n} - I)^T}{D^\beta }\Phi {D^\beta }{\Phi ^T}{D^\beta }({\gamma ^n}{({P^\pi })^n} - I)\Phi  = \Gamma  \succ 0$. Therefore,~\eqref{eq:objective1} is strongly convex, and the coefficient is given by $\mu = {\lambda _{\min }}(\Gamma )$~\cite[Thm.~2.1.11]{nesterov2018lectures}. This completes the proof.
\end{proof}

When strongly convex, the gradient descent algorithm in~\eqref{eq:GD-algo1} can find the unique global optimal solution $\bar \theta ^n$ of~\eqref{eq:opt2} with linear convergence rate, which also a stationary point satisfying~\eqref{eq:stationary2}.
Moreover, when strongly convex, the unique global optimal solution of~\eqref{eq:opt2} satisfying~\eqref{eq:stationary2} is identical to the unique fixed point, $\theta^n_*$, of $n$-PBE in~\eqref{eq:1} without~\cref{assumption:2}.
\begin{theorem}[Stationary point]
Suppose that $n \geq n^*$ so that $\Pi {T^n}$ is a contraction. Then, the unique stationary point, ${\bar \theta }^n$, satisfying~\eqref{eq:stationary2}, is the unique fixed point, $\theta^n_*$, of $n$-PBE in~\eqref{eq:1}.
\end{theorem}
\begin{proof}
The stationary point satisfies
\begin{align*}
{\nabla _\theta }f({{\bar \theta }^n}) =& {\Phi ^T}{({\gamma ^n}{({P^\pi })^n} - I)^T}{D^\beta }\Phi {D^\beta }{\Phi ^T}{D^\beta }\\
&\times ({T^n}(\Phi {{\bar \theta }^n}) - \Phi {{\bar \theta }^n})\\
=& 0.
\end{align*}

Since ${\Phi ^T}{({\gamma ^n}{({P^\pi })^n} - I)^T}D^\beta \Phi$ is nonsingular by~\cref{thm:contraction-property1}, the above equality is equivalent to
\begin{align*}
{\Phi ^T}{D^\beta }({T^n}(\Phi {{\bar \theta }^n}) - \Phi {{\bar \theta }^n}) = 0.
\end{align*}
Next, multiplying both sides of the above equality by $\Phi {({\Phi ^T}{D^\beta }\Phi )^{ - 1}}$ leads to $\Pi (\Pi {T^n}(\Phi {{\bar \theta }^n}) - \Phi {{\bar \theta }^n}) = \Pi {T^n}(\Phi {{\bar \theta }^n}) - \Phi {{\bar \theta }^n} = 0$, which is the fixed point equation of $n$-PBE. This completes the proof.
\end{proof}

\begin{theorem}[Lipschitz continuity]\label{thm:Lipschitz2}
Suppose that $n \geq n^*$ so that $\Pi {T^n}$ is a contraction. Then, the gradient in~\eqref{eq:gradient2} is $L$-Lipschitz continuous, i.e.,
\begin{align*}
{\left\| {{\nabla _\theta }f(x) - {\nabla _\theta }f(y)} \right\|_2} \le L{\left\| {x - y} \right\|_2},
\end{align*}
where $L = {\left\| \Gamma  \right\|_2}$.
\end{theorem}
\begin{proof}
Noting the gradient in~\eqref{eq:gradient2}, simple calculations lead to the desired conclusion.
\end{proof}

\begin{theorem}[{\cite[Thm.~2.1.15]{nesterov2018lectures}}]\label{thm:convergence4}
Suppose that $n \geq n^*$ so that $\Pi {T^n}$ is a contraction.
The gradient descent algorithm in~\eqref{eq:GD-algo1} with the constant step-size $\alpha  = \frac{2}{{\mu  + L}}$ satisfies
\begin{align*}
\left\| {{\theta _k} - \theta _*^n} \right\|_2^2 \le {\left( {\frac{{L - \mu }}{{L + \mu }}} \right)^{2k}}\left\| {{\theta _0} - \theta _*^n} \right\|_2^2,
\end{align*}
where $\mu$ and $L$ such that $L \ge \mu$ are defined in~\cref{thm:convexity2} and~\cref{thm:Lipschitz2}, respectively.
\end{theorem}
\begin{proof}
The proof follows that of~\cite[Thm.~2.1.15]{nesterov2018lectures}.
\end{proof}

\section{System operator}\label{sec:system}
Until now, we have studied model-based approaches, the classical dynamic programming and the gradient-based algorithms, to solve the policy evaluation problem. In this section, we will consider another class of model-based iterative algorithms based on the methods for solving general linear equations~\cite{kelley1995iterative}.
In particular, let us first consider the $n$-PBE again
\begin{align*}
{\Phi ^T}{D^\beta }{T^n}(\Phi \theta ) = {\Phi ^T}{D^\beta }\Phi \theta,
\end{align*}
which can be written as the following linear equation form:
\begin{align*}
&\underbrace {\,{\Phi ^T}{D^\beta }({R^\pi } + \gamma {P^\pi }{R^\pi } +  \cdots  + {\gamma ^{n - 1}}{{({P^\pi })}^{n - 1}}{R^\pi })}_{=:b}\\
=& \underbrace {{\Phi ^T}{D^\beta }(I - {\gamma ^n}{{({P^\pi })}^n})\Phi }_{=:A}\theta.
\end{align*}

We consider a Richardson type iteration~\cite{kelley1995iterative} of the form
\begin{align}
{\theta _{k + 1}} = {\theta _k} + \alpha {\Phi ^T}D^\beta({T^n}(\Phi {\theta _k}) - \Phi {\theta _k}),\label{eq:system1}
\end{align}
where $\alpha>0$ is a step-size. We will call the operator $F(x): = x + \alpha {\Phi ^T}D^\beta ({T^n}(\Phi x) - \Phi x)$ a system operator. Then,~\eqref{eq:system1} can be written as ${\theta _{k + 1}} = F({\theta _k})$. We can prove that the iterate $\theta_k$ converges to $\theta_*^n$ for a sufficiently large $n$ and sufficiently small $\alpha$. This result and related lemmas are given below.
\begin{lemma}\label{lemma:3}
There exists a positive integer $\bar n^*$ such that ${{\Phi ^T}D^\beta (I-{\gamma ^n}({P^\pi })^n)\Phi }$ becomes Hurwitz for any $n\ge \bar n^*$.
\end{lemma}
\begin{proof}
Since
\begin{align*}
&\mathop {\lim }\limits_{n \to \infty } \left[ {{\Phi ^T}D^\beta ({\gamma ^n}{{({P^\pi })}^n} - I)\Phi  + {\Phi ^T}{{({\gamma ^n}{{({P^\pi })}^n} - I)}^T}D^\beta \Phi } \right]\\
=& - 2{\Phi ^T}D^\beta \Phi  \prec 0,
\end{align*}
by continuity, there exists a positive integer $\bar n^*$ such that
\begin{align*}
{\Phi ^T}D^\beta ({\gamma ^n}{({P^\pi })^n} - I)\Phi  + {\Phi ^T}{({\gamma ^n}{({P^\pi })^n} - I)^T}D^\beta \Phi  \prec 0
\end{align*}
for all $n \geq \bar n^*$. This implies that ${{\Phi ^T}D^\beta (I-{\gamma ^n}({P^\pi })^n)\Phi }$ is Hurwitz stable.
\end{proof}
\begin{lemma}\label{lemma:4}
Suppose that the matrix $B$ is Hurwitz stable. Then, there exists a sufficiently small $\alpha^* >0 $ such that $A = I + \alpha B $ is Schur stable for all $\alpha \le \alpha^*$.
\end{lemma}
\begin{proof}
If $B$ is Hurwitz stable, then by the Lyapunov argument, there exists a Lyapunov matrix $P \succ 0$ such that ${B^T}P + PB =  - I$~\cite{chen1995linear}.
Next, with $A = I + \alpha B $, we have
\begin{align*}
{A^T}PA =& {(I + \alpha B)^T}P(I + \alpha B)\\
=& P + \alpha PB + \alpha {B^T}P + {\alpha ^2}{B^T}PB\\
=& P - \alpha I + {\alpha ^2}{B^T}PB.
\end{align*}
Then, it is clear that there exists a sufficiently small $\alpha^* >0 $ such that
\begin{align*}
{A^T}PA = P - \alpha I + {\alpha ^2}{B^T}PB \prec P,
\end{align*}
which implies that $A$ is Schur. This completes the proof.
\end{proof}

Based on the above two results, we are now ready to establish the convergence of the algorithm~\eqref{eq:system1}.
\begin{theorem}(Convergence)\label{thm:convergence3}
There exists a positive integer $\bar n^*$ and a positive real number $\alpha^*$ such that for any $n \geq \bar n^*$ and $\alpha \leq \alpha^*$, the iterate in~\eqref{eq:system1} converges to $\theta_*^n$.
\end{theorem}
\begin{proof}
Combining~\eqref{eq:system1} and the fixed point equation in~\eqref{eq:1}, it follows that
\begin{align*}
{\theta _{k + 1}} - \theta _*^n = \underbrace {(I + \alpha {\Phi ^T}D^\beta({\gamma ^n}{{({P^\pi })}^n} - I)\Phi )}_{ = :A}({\theta _k} - \theta _*^n),
\end{align*}
which is a discrete-time linear time-invariant system~\cite{chen1995linear}. Therefore, the convergence of~\eqref{eq:system1} is equivalent to the Schur stability of $A$.
Now, we will prove that $A$ is Schur stable if $n$ is sufficiently large and $\alpha$ is sufficiently small.
First of all, by~\cref{lemma:3}, there exists a positive integer $\bar n^*$ such that ${{\Phi ^T}D^\beta (I-{\gamma ^n}({P^\pi })^n)\Phi }$ becomes a Hurwitz stable matrix for any $n\ge \bar n^*$. Next, by~\cref{lemma:4}, there exists a sufficiently small $\alpha^* >0 $ such that $A$ is Schur stable. This completes the proof.
\end{proof}

In this section, we have proposed a different algorithm in~\eqref{eq:system1} from the classical dynamic programming in~\cref{sec:dynamic-programming} and the gradient descent methods in~\cref{sec:gradient1} and~\cref{sec:gradient2}, and analyzed its convergence based on the control system perspectives~\cite{chen1995linear}. All the iterative algorithms studied until now assume that the model is already known.
In the next section, we will study model-free reinforcement learning algorithms based on these algorithms.

\section{Off-policy multi-step TD-learning based on the system operator}

For convenience, in this paper, we consider the sampling oracle that takes the initial state $s_0$, and generates the sequences of states $({s_1},{s_2}, \ldots ,{s_n})$, actions $({a_0},{a_1}, \ldots ,{a_{n - 1}})$, and rewards $({r_1},{r_2}, \ldots ,{r_{n}})$ following the given constant behavior policy $\beta$.

The iterative algorithm in~\eqref{eq:system1} suggests an off-policy $n$-step TD-learning algorithm ($n$-TD) given in~\cref{algo:TD1}.
Note that~\cref{algo:TD1} can be viewed as a stochastic approximation of~\eqref{eq:system1} by replacing the model parameters by the corresponding samples of the state and action. Moreover, \cref{algo:TD1} can be viewed as a standard off-policy $n$-step TD-learning with the importance sampling method.s

It is also important to note that~\cref{algo:TD1} is introduced solely for conceptual purposes and not as a feasible alternative for practical use because it requires a sampling oracle that can generate the entire i.i.d. samples that are used at each iteration in $n$-TD. However, theoretical studies on~\cref{algo:TD1} may give some insights and help us develop more practical methods.

\begin{algorithm}[h]
\caption{Multi-step off-policy TD-learning}
\begin{algorithmic}[1]

\State Initialize $(\theta _0,\lambda_0 )$.

\For{iteration step $i \in \{0,1,\ldots\}$}

\State Sample $s_0 \sim d^{\beta}$, and sample $({s_1},{s_2}, \ldots ,{s_n})$, $({a_0},{a_1}, \ldots ,{a_{n - 1}})$, and $({r_1},{r_2}, \ldots ,{r_{n}})$ using the sampling oracle.

\State Update parameters according to
\begin{align*}
{\theta _{i + 1}} = {\theta _i} + {\alpha _i}{\rho _{n - 1}}(G - {V_{{\theta _i}}}({s_0})){\varphi(s_0)},
\end{align*}
where $\rho_{n-1} : = \prod\nolimits_{k = 0}^{n - 1} {\frac{{\pi ({a_k}|{s_k})}}{{\beta ({a_k}|{s_k})}}}$ is the importance sampling ratio, $\varphi (s) = {\Phi ^T}{e_s}$ is the $s$-th row vector of $\Phi$, $G = \sum\limits_{k = 0}^{n - 1} {\gamma ^k r_{k+1}}  + {\gamma ^n}{V_{{\theta _i}}}({s_n})$, and ${V_{{\theta _i}}}(s) = e_s^T\Phi {\theta _i}$.

\EndFor
\end{algorithmic}
\label{algo:TD1}
\end{algorithm}

Following the ideas in~\cite{borkar2000ode}, the convergence of~\cref{algo:TD1} can be easily established.
\begin{theorem}\label{thm:convergence5}
Consider~\cref{algo:TD1}, and assume that the step-size satisfy
\begin{align}
&\alpha_k>0,\quad \sum_{k=0}^\infty {\alpha_k}=\infty,\quad \sum_{k=0}^\infty{\alpha_k^2}<\infty.\label{eq:step-size-rule}
\end{align}

Then, $\theta_k \to \theta_*^n$ as $k \to \infty$ with probability one for any $n \ge \bar n^*$, where $\bar n^*$ is given in the statement of~\cref{lemma:3}.
\end{theorem}
\begin{proof}
The so-called O.D.E. model of~\cref{algo:TD1} is
\begin{align}
{{\dot \theta }_t} = {\Phi ^T}{D^\beta }({T^n}(\Phi {\theta _t}) - \Phi {\theta _t}). \label{eq:ODE1}
\end{align}
By~\cref{lemma:3}, for $n \ge \bar n^*$, ${{\Phi ^T}D^\beta (I-{\gamma ^n}({P^\pi })^n)\Phi }$ is Hurwitz, and hence,~\eqref{eq:ODE1} is globally asymptotically stable. Then, the proof is completed by using Borkar and Mayen theorem in~\cite[Thm.~2.2]{borkar2000ode}.
\end{proof}

\cref{thm:convergence5} tells us that $n$-TD can solve the policy evaluation problem with a sufficiently large $n$. In other words, it can resolve the deadly triad problem.

\section{Off-policy $n$-step TD-learning based on the gradient operator}

In the previous section, an off-policy $n$-TD has been considered.
In this section, we will consider an off-policy $n$-step GTD algorithm ($n$-GTD).
To derive it, we follow similar steps as in~\cite{lee2023new}. In particular, let us consider the optimization problem~\eqref{eq:opt2}, which can be reformulated as the constrained optimization
\begin{align}
{\min _{\theta  \in {\mathbb R}^m}}\quad {\rm{s}}{\rm{.t}}{\rm{.}}\quad 0 = {\Phi ^T}{D^\beta }({T^n}(\Phi {\theta}) - \Phi \theta).\label{eq:opt3}
\end{align}

Note that in~\eqref{eq:opt3}, we introduce a null objective, $f\equiv 0$, to fit the problem into an optimization form. We can easily prove that the optimization admits a unique solution~\cite{lee2023new}, which is identical to the solution of $n$-PBE~\eqref{eq:1}.

Next, we formulate~\eqref{eq:opt3} into a min-max saddle-point problem by introducing the corresponding Lagrangian function
\begin{align}
{\min _{\theta  \in {\mathbb R}^m}}{\max _{\lambda  \in {\mathbb R}^m}}L(\theta ,\lambda ): = {\lambda ^T}{\Phi ^T}{D^\beta }({T^n}(\Phi \theta ) - \Phi \theta ).\label{eq:saddle1}
\end{align}

Now, a regularization term is introduced to make it strongly concave in $\lambda$, and obtain the following modification:
\begin{align}
L(\theta ,\lambda ): = {\lambda ^T}{\Phi ^T}{D^\beta }({T^n}(\Phi \theta ) - \Phi \theta ) - \frac{1}{2}{\lambda ^T}{\Phi ^T}{D^\beta }\Phi \lambda. \label{eq:Lagrangian1}
\end{align}

The corresponding saddle-point problem of~\eqref{eq:Lagrangian1} is then given as follows.
\begin{align}
{\min _{\theta  \in {\mathbb R}^m}}{\max _{\lambda  \in {\mathbb R}^m}}L(\theta ,\lambda ):=& {\lambda ^T}{\Phi ^T}{D^\beta }({T^n}(\Phi \theta ) - \Phi \theta )\nonumber\\
& - \frac{1}{2}{\lambda ^T}{\Phi ^T}{D^\beta }\Phi \lambda. \label{eq:saddle2}
\end{align}

We can prove that the solutions of~\eqref{eq:saddle2} is identical to the solution of $n$-PBE in~\eqref{eq:1}.
\begin{proposition}[{\cite{lee2023new}}]\label{prop:2}
A solution of~\eqref{eq:saddle2} exists, is unique, and is given by $\theta = \theta_*^n$ and $\lambda = 0$.
\end{proposition}

Now, let us turn our attention to the so-called continuous-time primal-dual gradient dynamics~\cite{qu2018exponential}
\begin{align}
{{\dot \theta }_t} = & - {\nabla _\theta }L({\theta _t},{\lambda _t}) =  - {({\gamma ^n}{({P^\pi })^n}\Phi  - \Phi )^T}{D^\beta }\Phi {\lambda _t}\nonumber\\
{{\dot \lambda }_t} =& {\nabla _\lambda }L(\theta ,{\lambda _t}) = {\Phi ^T}{D^\beta }({T^n}(\Phi {\theta _t}) - \Phi {\theta _t} - \Phi {\lambda _t}).\label{eq:7}
\end{align}

By replacing the updates by those with stochastic approximations of the model parameters using samples of the state and action, one can obtain a GTD version summarized in~\cref{algo:TD2}.
\begin{algorithm}[h]
\caption{Multi-step off-policy gradient TD-learning}
\begin{algorithmic}[1]

\State Initialize $(\theta _0,\lambda_0 )$.

\For{iteration step $i \in \{0,\ldots\}$}

\State Sample $s_0 \sim d^{\beta}$, and sample $({s_1},{s_2}, \ldots ,{s_n})$, $({a_0},{a_1}, \ldots ,{a_{n - 1}})$, and $({r_0},{r_1}, \ldots ,{r_{n - 1}})$ using the sampling oracle.

\State Update parameters according to
\begin{align*}
{\theta _{i + 1}} =& {\theta _i} + {\alpha _i}{\rho _{n - 1}}(\varphi ({s_0}) - {\gamma ^n}\varphi ({s_n})){H_{{\lambda _i}}}({s_0})\\
{\lambda _{i + 1}} =& {\lambda _i} + {\alpha _i}{\rho _{n - 1}}(G - {V_{{\theta _i}}}({s_0}) - {H_{{\lambda _i}}}({s_0}))\varphi ({s_0}),
\end{align*}
where $\rho_{n-1}: = \prod\nolimits_{k = 0}^{n - 1} {\frac{{\pi ({a_k}|{s_k})}}{{\beta ({a_k}|{s_k})}}}$ is the importance sampling ratio, $\varphi (s) = {\Phi ^T}{e_s}$ is the $s$-th row vector of $\Phi$, $G = \sum\limits_{k = 0}^{n - 1} {\gamma ^k  r_{k+1}}  + {\gamma ^n}{V_{{\theta _i}}}({s_n})$, ${V_{{\theta _i}}}(s) = e_s^T\Phi {\theta _i}$, and ${H_{{\lambda _i}}}(s) = e_s^T\Phi {\lambda _i}$.

\EndFor
\end{algorithmic}
\label{algo:TD2}
\end{algorithm}

Similar to~\cref{algo:TD1}, the convergence of~\cref{algo:TD2} can be easily established.
\begin{theorem}
Consider~\cref{algo:TD2}, and assume that the step-size satisfy~\eqref{eq:step-size-rule}.
Then, $\theta_k \to \theta_*^n$ as $k \to \infty$ with probability one for any $n \ge \bar n^*$, where $\bar n^*$ is given in the statement of~\cref{lemma:3}.
\end{theorem}
\begin{proof}
The O.D.E. model of~\cref{algo:TD2} is given in~\eqref{eq:7}. The global asymptotic stability of~\eqref{eq:7} can be proved following the steps in~\cite{lee2023new} and using the notion of the primal-dual gradient dynamics~\cite{qu2018exponential}.
Then, the proof is completed by using Borkar and Mayen theorem in~\cite[Thm.~2.2]{borkar2000ode}.
\end{proof}

It is well known that the off-policy GTD algorithms~\cite{sutton2009convergent,sutton2009fast,ghiassian2020gradient,lee2023new} guarantee convergence under a deadly triad of function approximation, bootstrapping, and off-policy learning.
However, they require~\cref{assumption:2} for convergence, which may not be satisfied in general.
On the other hand, a benefit of using $n$-GTD in~\cref{algo:TD2} compared to the standard GTDs is that it does not require~\cref{assumption:2} provided that $n \ge \bar n^*$.

\section{Conclusion}

In this paper, we have investigated the convergence and properties of $n$-step TD-learning algorithms. We have proved that under the deadly triad scenario, the $n$-step TD-learning algorithms converge to useful solutions as the sampling horizon $n$ increases sufficiently. We have comprehensively examined the fundamental properties of their model-based deterministic counterparts, which can be viewed as prototype deterministic algorithms whose analysis plays a pivotal role in understanding and developing their model-free RL counterparts. Based on the analysis and insights from the deterministic algorithms, we have established convergence of two $n$-step TD-learning algorithms.

\bibliographystyle{IEEEtran}
\bibliography{reference}

\end{document}